\documentclass{IEEEtran}

\usepackage{amssymb}
\usepackage{amsthm}
\usepackage{cite}
\usepackage[multiple]{footmisc}
\usepackage{mathtools}
\usepackage{enumitem}
\usepackage{color}
\mathtoolsset{showonlyrefs}

\newcommand{\R}{\mathbb{R}}

\renewcommand{\P}{\mathbb{P}}

\newcommand{\ones}{{\bf 1}}
\newcommand{\myDiag}{\textup{diag}\,}

\newcommand{\BRhat}{\widehat{\text{BR}}{}^\lambda}
\newcommand{\ND}{\text{ND}}

   \def\vx{{\bf x}}

\newtheorem{theorem}{Theorem}

\newtheorem{proposition}[theorem]{Proposition}
\newtheorem{remark}[theorem]{Remark}

\newtheorem{definition}[theorem]{Definition}

\title{Smooth Fictitious Play in $N\times 2$ Potential Games}
\author{Brian Swenson and H. Vincent Poor
\thanks{\noindent This work was partially supported by the Air Force Office of Scientific Research under MURI Grant FA9550-18-1-0502. %and was partially supported by the National Science Foundation under Award Number CCF 1513936
\newline
The authors are with the Department of Electrical Engineering, Princeton University, Princeton, NJ 08540 (bswenson@princeton.edu and poor@princeton.edu),\newline
}}

\begin{document}

\maketitle

% REQUIRED
\begin{abstract}
The paper shows that smooth fictitious play converges to a neighborhood of a pure-strategy Nash equilibrium with probability 1 in almost all $N\times 2$ ($N$-player, two-action) potential games. The neighborhood of convergence may be made arbitrarily small by taking the smoothing parameter to zero. Simple proof techniques are furnished by considering \emph{regular potential games}.
\end{abstract}

% REQUIRED

\section{Introduction}
%This paper considers game-theoretic learning in potential games.
Potential games, originally introduced in \cite{monderer1996potential}, have widespread application across economics, engineering, and computer science \cite{marden2009cooperative,stankovic2011distributed,sandholm2001potential}.
In a potential game, there exists a potential function which all players implicity wish to optimize. The existence of a potential function has many benefits, not least of which is that it guarantees the existence of pure Nash equilibria (NE).
%In engineering applications, the popularity of potential games stems from their applicability in decentralized control applications \cite{marden2009cooperative}.
%An important property of potential games
%%(that is particularly advantageous in decentralized control applications)
%is that they guarantee the existence of pure Nash equilibria (NE).
Pure NE are deterministic, stable, and maximize the potential function (locally). Consequently, in game-theoretic learning applications (particularly in those related to decentralized control), pure NE are typically vastly preferable to their mixed counterparts.% as limit points of the algorithm.

Another important benefit of potential games is that they are amenable to game-theoretic learning processes. The potential function serves as a Lyapunov function that ensures that virtually any reasonable learning algorithm will converge to the set of NE. However, many game-theoretic learning algorithms operate by evolving in the space of mixed strategies, and for these algorithms it is generally not clear if or when the limit point of an algorithm will be a pure vs mixed NE.
%for learning algorithms where
%that rely on mixed-strategy adaptation (i.e.,
%learning evolves in the space of mixed strategy,
%there are few results clarifying if or when the limit point of an algorithm will be a pure vs mixed NE.

Common examples of algorithms relying on mixed-strategy adaptation include
the exponential weights algorithm \cite{freund1999adaptive,heliou2017learning}, regret matching \cite{hart2000simple}, actor-critic algorithms \cite{leslie2006generalised}, gradient-descent based algorithms \cite{shamma2005dynamic}, fictitious-play based algorithms \cite{shamma2005dynamic,lambert2005fictitious,swenson2017single,
swenson2015empirical,swenson2012distributed,swenson2015asynchronous,swenson2017exponential,swenson2018distributed,swenson2017robustness}, or, in continuous-time, Brown-von Neumann-Nash dynamics \cite{hofbauer2003evolutionary}, replicator dynamics \cite{hofbauer2003evolutionary} and best-response dynamics \cite{matsui1992best,swenson2018best,swenson2018Asilomar}. In such algorithms, convergence to mixed NE in potential games can be highly problematic.

In practice, there is a general understanding among practitioners that most reasonable learning dynamics \emph{ought} to converge to pure NE in potential games \cite{arslan2007autonomous}. However, a rigorous understanding of this issue, in general, is lacking.
%one typically does not expect learning algorithms to converge to mixed NE potential games.

%In practice, one typically does not expect convergence to mixed NE potential games. However, a rigorous understanding of this issue has been elusive.
In this paper we address this deficiency for the special case of \emph{smooth fictitious play} (smooth FP) \cite{fudenberg1993learning,fudenberg1995consistency,
benaim2013consistency,hernandez2014selective} in $N\times 2$ potential games. Fictitious play (FP), introduced in \cite{brown1951iterative}, is a canonical algorithm for learning in games. Smooth FP is a stochastic variant of FP that has been shown to achieve no-regret learning and is useful for multi-agent online learning \cite{fudenberg1995consistency,roughgarden2018lectures}.
In potential games, smooth FP is known to converge to a neighborhood of the set of NE \cite{hofbauer2002global}. However, the issue of showing generic convergence to pure NE (or a neighborhood thereof) has not yet been resolved.

As the main result of this paper, we show that in almost all $N\times 2$ potential games, smooth FP converges to the neighborhood of a pure-strategy NE with probability 1. (See Theorem \ref{thrm:main-result} for our main result and Section \ref{sec:regular-games} for a definition of ``almost all potential games'')

We are able to obtain a simple proof of this result by considering the notion of a  \emph{regular potential game} \cite{swenson2017regular} (see also \cite{harsanyi1973oddness,van1991stability}). In a regular game, all equilibria are ``nondegenerate'' in an appropriate sense that makes them well suited for game-theoretic learning applications (see Section \ref{sec:regular-games} for more details). This allows for a simple proof of our main result within the class of regular potential games.

The recent work \cite{swenson2017regular} showed that almost all potential games are regular (see also Section \ref{sec:regular-games} below). Thus, using \cite{swenson2017regular}, the results we derive here for smooth FP in regular potential games immediately extend to almost all potential games. This is the key idea of our approach.

We remark that regular potential games were used to study continuous-time best-response (CT-BR) dynamics in \cite{swenson2018best}. CT-BR dynamics may be viewed as the continuous-time version of standard fictitious play. In \cite{swenson2018best} it was shown that CT-BR dynamics converge to pure NE from almost all initial conditions in almost all potential games. Our work here shows an analogous result for smooth FP; however, the proof strategy is much simpler in this case due to the smoothness of the underlying dynamical system.

We also remark that, while we focus on games, the algorithms we consider fall under the general umbrella of dynamical systems for optimizing discrete problems, e.g., \cite{absil2004continuous,vidyasagar1993minimum,lambert2005fictitious}, and are applicable in this domain, for example, as parallel processing algorithms.

%deal with continuum approaches for optimizing discrete problems, and fall under the general umbrella of 'dynamical systems that optimize over the hypercube'

The remainder of the paper is organized as follows. Section \ref{sec:notation} sets up notation. Section \ref{sec:smooth-FP} introduces smooth FP. Section \ref{sec:regular-games} introduces regular potential games and the notion of ``almost all games'' in this context. Section \ref{sec:ND-and-NE} elucidates the relationship between the limit points of smooth FP and the set of NE and presents our main result (see Theorem \ref{thrm:main-result}). Section \ref{sec:conclusion} concludes the paper.

\section{Notation}\label{sec:notation}
A normal form game is given by the tuple $\Gamma = (N,(A_i,u_i)_{i=1}^N)$, where $N$ is the number of players, $A_i=\{a_i^1,\ldots,a_i^{K_i}\}$ is the action set of player $i$ (assumed to be finite) and $u_i:A_1\times\cdots\times A_N\to\R$ is the utility function of player $i$. Given $i\in\{1,\ldots,N\}$ we let $A_{-i} = \prod_{j\not= i} A_j$.

The main results of this paper will focus on $N\times 2$ games; i.e., games with $N$ players and two actions per player.
%\begin{definition}
%$\Gamma$ is said to be an $N\times 2$ game if $\Gamma$ has $N$ players with 2 actions per player.
%\end{definition}
Unless otherwise stated, we will assume through the remainder of the paper that we are dealing only with $N\times 2$ games.\footnote{We note however, that with the exception of Proposition \ref{prop:ND-hyperbolic} and Theorem \ref{thrm:main-result}, all results are valid for games of arbitrary size.}

A game is said to be a potential game, as introduced in \cite{monderer1996potential}, if there exists a function $u:A_1\times\cdots\times A_N\to\R$ such that
$u(a_i,a_{-i}) - u(a_i',a_{-i}) = u_i(a_i,a_{-i}) - u_i(a_i',a_{-i})$ for all $a_i,a_i'\in A_i$, and $a_{-i}\in A_{-i}$.
In this paper we will focus on games with identical interests, i.e., games where $u_i = u$ for all $i$, as these games are strategically equivalent to potential games and the dynamical systems considered in this paper behave identically on a potential game or an associated identical-interests game.

To properly define smooth fictitious play we must consider the set of \emph{mixed strategies} where players mix probabilistically between actions. Let $\Delta_i$ denote the set of probability distributions over the action set $A_i$. Since we assume that $\Gamma$ is an $N\times 2$ game, in an abuse of notation we will represent a strategy $x_i\in \Delta_i$ as single real number $x_i\in [0,1]$, which is interpreted as the probability mass placed on the action $a_i^1$.
The set of \emph{joint} mixed strategies is given by $\Delta := \Delta_1\times\cdots\times \Delta_N$, which may be viewed as the hypercube $\Delta = [0,1]^N$.

Given a mixed strategy $x\in \Delta$, we define
$$
U(x) := \sum_{\substack{i=1,\ldots,N\\ k_i=1,2}} z_1^{k_1}(x_1)\cdots z_N^{k_N}(x_N) u(a_1^{k_1},\ldots,a_N^{k_N})
$$
where $z_i^1(x_i) = x_i$ represents the probability of player $i$ playing action $a_i^1$ and $z_i^2 = (1-x_i)$ represents the complementary probability. We emphasize that $U(x)$ is simply the expected value of the potential given the mixed strategy $x\in \Delta$.
%When we refer to the ``potential function'' of a game, we will generally mean $U$ defined above.
In an abuse of notation, we will use $U(a_i^k,x_{-i})$ to represent the expected potential of action $a_i^k\in A_i$ under the mixed strategy $x_{-i}\in \prod_{j\not= i} \Delta_j$.

The set of Nash equilibria is given by
$$
NE := \{x\in \Delta: U(x_i,x_{-i}) \geq U(x_i',x_{-i}) \text{ for all } x_i' \in \Delta\}
$$

Smooth fictitious is defined using the notion of a smoothed best response. Formally, the smoothed (or logit) best response is given by
\begin{equation} \label{eq:BR-def}
\BRhat_i(x) := \frac{\exp(\frac{1}{\lambda} U(a_i^1,x_{-i}))}{\sum_{k=1,2} \exp(\frac{1}{\lambda} U(a_i^k,x_{-i}))},
\end{equation}
for smoothing parameter $\lambda>0$.

We remark that, since we focus on $N\times 2$ games, we have simplified our notation. In general, $\BRhat_i(x)$ is an element of $\Delta_i$ specifying the weight placed on each action in $A_i$. Since we treat $\Delta_i$ to be $[0,1]$ here, $\BRhat_i(x)$ should be interpreted as the weight placed on $a_i^1$. The weight placed on $a_i^2$ is simply obtained as $1-\BRhat_i(x)$.

Note that as $\lambda\to 0$ the perturbed best response approaches a probability distribution uniformly distributing its mass on actions that maximize the potential function.
Similarly, we define the joint smoothed best response as
$$
\BRhat(x) := \BRhat_1(x)\times\cdots\times \BRhat_N(x).
$$
We remark that for all $x\in\Delta$, $\BRhat(x)$ is unique and lies in the interior of $\Delta$.

Following \cite{fudenberg1998theory} we refer to a fixed point of $\BRhat$ as a \emph{Nash distribution} (with parameter $\lambda$), and denote the set of Nash distributions by
\begin{equation} \label{def:ND}
\ND(\lambda) := \{x:x=\BRhat(x)\}.
\end{equation}
The set $\ND(\lambda)$, $\lambda>0$ is not the same as the set of NE. However, we will see that in regular potential games, convergence of the set of Nash distributions to the set of NE does occur as $\lambda\to 0$ (see Theorem \ref{thrm:ND-to-NE} below).

%In general, we will only expect smooth FP to converge to Nash distributions that are stable, as defined below.

Finally, as a matter of notation when we say that a function is of class $C^k$, $k\geq 1$ we mean that it is $k$ times continuously differentiable. Given a function $F:\R^d\to\R^d$, we use the notation $DF(x)$ to denote the Jacobian of $F$ at $x$. Given a set of scalars $a_1,\ldots,a_k$, we let $\myDiag(a_1,\ldots,a_k)$ denote the $k\times k$ matrix with entries $a_1,\ldots,a_k$ on the diagonal.

\section{Smooth FP} \label{sec:smooth-FP}

\subsection{Smooth fictitious Play}
Suppose players repeatedly play some fixed game $\Gamma$. For $n\geq 0$, let $a_i(n)$ denote the action taken by player $i$ in round $n$, and let $x_i(n) = \frac{1}{n}\sum_{s=1}^n \ones_{\{a_i(s) = a_i^1\}}$ represent the \emph{empirical distribution} of player $i$, where $\ones_{\{a_i(s) = a_i^1\}} = 1$ if $a_i(s) = a_i^1$ and $\ones_{\{a_i(s) = a_i^1\}} = 0$ otherwise.

Smooth fictitious play is defined as follows.
Let the initial action $a_i(0)$ be chosen arbitrarily. For $n\geq 0$, players choose their next-stage action according to the probabilistic rule
\begin{equation} \label{eq:main-algo-1}
\P(a_i(n+1) = a_i^k) = \frac{\exp(\frac{1}{\lambda} U(a_i^k,x_{-i}(n)))}{\sum_{m=1}^{2} \exp(\frac{1}{\lambda} U(a_i^k,x_{-i}(n)))},\\
\end{equation}
or equivalently, $\P(a_i(n+1) = a_i^1) = \BRhat_i(x(n))$ and $\P(a_i(n+1) = a_i^2) = 1-\BRhat_i(x(n))$.
The empirical distribution is updated as
\begin{equation} \label{eq:main-algo-2}
x_i(n+1) = x_i(n) + \frac{1}{n+1}\left( a_i(n+1) - x_i(n) \right).
\end{equation}

We refer to the update procedure \eqref{eq:main-algo-1}--\eqref{eq:main-algo-2} as smooth FP.\footnote{More general notions of smooth FP are considered in \cite{fudenberg1995consistency}.} Smooth FP is known to converge to the set $\ND(\lambda)$ in several classes of games including potential games \cite{hofbauer2002global}.

\subsection{Smoothed Best-Response Dynamics}
The long run behavior of the state in smooth FP (considered to be the empirical distribution \eqref{eq:main-algo-2}) is determined by the ODE
\begin{equation} \label{eq:BR-perturbed}
\dot \vx(t) = \BRhat(\vx(t)) - \vx(t),
\end{equation}
where $\vx:[0,\infty) \to \R^N$.

\begin{remark}
We remark that, in contrast to $\BRhat$, the standard best response (which is approximated by $\BRhat$ as $\lambda\to 0$)
is set valued and discontinuous. The properties of genuine best response dynamics in potential games are studied in \cite{swenson2018best}.
\end{remark}

The following definition, standard from dynamical systems theory, will be central to our treatment.
\begin{definition}[Hyperbolic rest point]
Consider a differential equation
\begin{equation} \label{eq:generic-ode}
\dot \vx = F(\vx)
\end{equation}
where $F:\R^n\to \R^n$ is $C^1$. A rest point $x$ of \eqref{eq:generic-ode}
%, i.e., a point $x$ where $F(x) = 0$,
is said to be \emph{hyperbolic} if the Jacobian $DF(x)$ is nonsingular.
\end{definition}

We will see that smooth FP will only converge to Nash distributions that are stable under \eqref{eq:BR-perturbed}, as defined below.
\begin{definition} [Linearly stable point]
We say that a rest point of \eqref{eq:generic-ode} is linearly stable if all eigenvalues of $DF(x)$ have negative real part.
%In particular, a Nash distribution $x\in\ND(\lambda)$ is linearly stable if all eigenvalues of $D\BRhat(x) - I$ have negative real part.
\end{definition}

The following theorem from \cite{hofbauer2002global} characterizes the limit points of \eqref{eq:main-algo-1}--\eqref{eq:main-algo-2} in terms of the rest points of \eqref{eq:BR-perturbed}.
\begin{theorem} \label{thrm:sandholm}
If all rest points of \eqref{eq:BR-perturbed} are hyperbolic, then with probability 1, the set of limit points of \eqref{eq:main-algo-1}--\eqref{eq:main-algo-2} is the set of linearly stable Nash distributions.
\end{theorem}

\section{Regular Potential Games} \label{sec:regular-games}
The notion of a regular game was first introduced by Harsanyi in \cite{harsanyi1973oddness}. The main advantage of regular games is that their equilibria are robust, nondegenerate,\footnote{Nash equilibria may possess many ``quirky'' properties, including discontinuity with respect to payoffs, connected sets of NE, and others \cite{van1991stability}. We use the term ``degenerate'' imprecisely here to refer generally to such behavior. See \cite{van1991stability} for a detailed treatment.} and easy to analyze \cite{van1991stability}.

Let $x^*$ be an equilibrium of a potential game. Without loss of generality, assume that the pure strategy set is ordered so that $x_i^*> 0$ for all $i$.

An equilibrium of a potential game is regular if it satisfies two properties. First, a regular equilibrium must be quasi-strict as defined below (see also \cite{van1991stability}).\footnote{We remark that this definition has been adapted to suit $N\times 2$ games.} Second, at a regular equilibrium, the Hessian of the potential function must be ``nondegenerate.''

We now define these properties formally.
\begin{definition}
An equilibrium $x^*$ is said to be quasi strict if $x_i^* = 1$ implies that $y_i^2$ is not a pure-strategy best response to $x^*_{-i}$.
\end{definition}
Before defining the notion of nondegeneracy for the Hessian, note that an equilibrium $x^*$ may be on the boundary of the strategy space. Informally, to characterize degeneracy, we need to check the behavior of derivatives only in coordinates where the constraints are not active.

To formalize this, let $\tilde N$ denote the number of \emph{mixing} players, i.e., $\tilde N = |\{ i\in\{1,\ldots,N\}:x_i^{*,1}\not = 1\}|$, and without loss of generality, assume that the strategy set is ordered so that $x_i^*<1$, $i=1,\ldots,\tilde N$ (i.e., players $1,\ldots,\tilde N$ used genuine mixed strategies and the remaining players use pure strategies at $x^*$).
The \emph{restricted} Hessian of $U$, relative to $x^*$, is given by\footnote{We note that, in this definition, we define $H(\cdot)$ with respect to some NE $x^*$ but allow for $H(\cdot)$ to be evaluated at an arbitrary $x\in \Delta$. This is so we may evaluate $H(\cdot)$ at elements of $\ND(\lambda)$ that some NE.}
\begin{equation} \label{eq:def-H}
H(x) := \left(\frac{\partial U(x)}{\partial x_i\partial x_j} \right)_{i,j=1,\ldots,\tilde N}
\end{equation}
%If $x^*$ is an interior equilibrium, then this is simply the traditional definition of the Hessian of $U$ at $x^*$. However, NE may occur on the boundary of the strategy space. In this case, the above definition only looks at derivatives in \emph{unconstrained} directions.

The notion of a regular equilibrium is now defined below.
\begin{definition}
A Nash equilibrium $x^*$ of a potential game is said to be \emph{regular} if
\begin{itemize}
\item [(i)]$x^*$ is quasi strict, and
\item [(ii)] The restricted Hessian $H(x^*)$ is invertible.
\end{itemize}
\end{definition}
\begin{remark}
We comment on the two extreme cases in the definition of regularity. Note that if an equilibrium $x^*$ is in the interior of the strategy space, then regularity simply reduces to the condition that $x^*$ is a nondegenerate critical point of $U$ in the standard sense. On the other hand, if $x^*$ is at a vertex of the strategy space, then regularity is equivalent to $x^*$ being a strict equilibrium (i.e., $u(a) > u(a_i',a_{-i})$ for all $a_i'\in A_i$ and all $i$). In the intermediate case that $x^*$ lies on a boundary of the simplex but is not a vertex, regularity may be seen as a mixture of these two conditions.
\end{remark}

We define the notion of a regular \emph{game} as follows.
\begin{definition}
A potential game $\Gamma$ is said to be regular if all equilibria in the game are regular.
\end{definition}

Regular games possess a multitude of desirable stability and robustness properties. See \cite{van1991stability} for an extensive treatment.

We would like to be able to say that ``almost all'' potential games are regular. To this end, we will now define a suitable notion of ``almost all'' in this context.
Suppose that we are given integers $N$ and $K_i$, $i=1,\ldots,N$. Consider the set of all $N$-player potential games having action spaces with cardinality $|A_i| = K_i$. Observe that any such game is uniquely defined by a payoff vector $u\in \R^{K_1\times\cdots\times K_N}$.

We say that almost all $N\times 2$ potential games satisfy a certain property if for any $N$, the set of all $N\times 2$ potential games for which the property fails to hold is a closed, measure-zero subset of $\R^{2N}$.

The following theorem from \cite{swenson2017regular} establishes that regularity is in fact a generic property within the class of potential games.
\begin{theorem}[ \hspace{-.6em} \cite{swenson2017regular}, Theorem 1] \label{thrm:reg-pot-games}
Almost all potential games are regular.
\end{theorem}

\section{Nash Distributions in \\ Regular Potential Games} \label{sec:ND-and-NE}
\subsection{Nash Distributions and Nash Equilibria}
In regular potential games the set of NE is finite (\hspace{-.01em}\cite{swenson2017regular}, Theorem 2). The following theorem establishes the close relationship between the set of NE and the set $\ND(\lambda)$ in regular potential games.
%that in such games, the set $\ND(\lambda)$ possesses the same cardinality as the set of NE. Moreover, there is a one-to-one correspondence between the set of NE and $ND(\lambda)$ with
\begin{theorem} \label{thrm:ND-to-NE}
Suppose $\Gamma$ is a regular potential game. Then the set of NE is finite. Moreover, for $\lambda>0$ sufficiently small there is a one-to-one correspondence between the set of NE and $\ND(\lambda)$ with each element of $\ND(\lambda)$ converging continuously to an associated Nash equilibrium point as $\lambda \to 0$.
\end{theorem}
We remark that, implicit in the above theorem is the fact that $\ND(\lambda)$ is finite for $\lambda$ sufficiently small (see also \cite{hofbauer2002global}, Theorem 4.3 and Theorem \ref{prop:ND-hyperbolic} below).
A proof of this result is omitted for brevity. However, the result follows readily by observing that the smoothed best response $\BRhat_i(x)$ is obtained by maximizing $U(x_i,x_{-i}) + \lambda\sum_{i=1}^n x_i\log(x_i)$. As $\lambda\to 0$, the critical points of this perturbed function converge to the set of NE in regular games.

By construction, all Nash distributions lie in the interior of the strategy set. However, in an abuse of terminology, we will use the following nomenclature.
\begin{definition}
We say that a Nash distribution is a \emph{pure-strategy Nash distribution} if it converges to a pure-strategy Nash equilibrium as $\lambda\to 0$.
\end{definition}

\subsection{Hyperbolicity of Nash Distributions}
Theorem \ref{thrm:sandholm} shows that in games with hyperbolic Nash distributions, the limit points of smooth FP are, almost surely, linearly stable Nash distributions. The following theorem shows that in regular potential games, all Nash distributions are hyperbolic and only pure-strategy Nash distributions are linearly stable.
Thus, together with Theorem \ref{thrm:sandholm} and Theorem \ref{thrm:reg-pot-games}, the following proposition will immediately imply the main result of the paper (Theorem \ref{thrm:main-result} below).
\begin{proposition} \label{prop:ND-hyperbolic}
If $\Gamma$ is a regular $N\times 2$ potential game, then for all $\lambda>0$ sufficiently small, all Nash distributions are hyperbolic. Moreover, a Nash distribution is linearly stable if and only if it is a pure-strategy Nash distribution.
\end{proposition}
\begin{proof}
Let
$$F_i(x) := \BRhat_i(x) - x_i,$$
and let $F(x) = (F_i(x))_{i=1}^N$ represent the right hand side of \eqref{eq:BR-perturbed}.
Note that we may express $U(x)$ as $U(x_i,x_{-i}) = x_iU(a_i^1,x_{-i}) + (1-x_i)U(a_i^2,x_{-i})$. Hence,
$$
\frac{\partial U(x)}{\partial x_i} = U(a_i^1,x_{-i}) - U(a_i^2,x_{-i}),
$$
and
\begin{equation} \label{eq:F-to-partialU}
\frac{\partial^2 U(x)}{\partial x_i\partial x_j} = \frac{\partial U(a_i^1,x_{-i})}{\partial x_j} - \frac{\partial U(a_i^2,x_{-i})}{\partial x_j}.
\end{equation}

We now compute $\frac{\partial \BRhat(x)}{\partial x_j}$. For $j=i$ we have $\frac{\partial \BRhat(x)}{\partial x_j} = 0 = \frac{\partial^2 U(x)}{\partial x_i^2}$. For $j\not = i$ we have
\begin{align}
\frac{\BRhat_i(x)}{\partial x_j} = &
\frac{1}{\lambda}\frac{\exp(\frac{1}{\lambda} U(a_i^k,x_{-i}))}{\sum_{k=1,2} \exp(\frac{1}{\lambda} U(a_i^k,x_{-i}))}\frac{\partial U(a_i^1,x_{-i})}{\partial x_j}\\
& - \frac{1}{\lambda}
\frac{\exp(\frac{1}{\lambda} U(a_i^k,x_{-i}))}{\left(\sum_{k=1,2} \exp(\frac{1}{\lambda} U(a_i^k,x_{-i}))\right)^2}\\
& \times \bigg(\exp(\frac{1}{\lambda} U(a_i^1,x_{-i}))\frac{\partial U(a_i^1,x_{-i})}{\partial x_j}\\
& \quad+ \exp(\frac{1}{\lambda} U(a_i^2,x_{-i}))\frac{\partial U(a_i^2,x_{-i})}{\partial x_j} \bigg)
\end{align}
Suppose henceforth that $x^\lambda\in \ND(\lambda)$, so that $x^\lambda = \BRhat(x^\lambda)$, and that $x^*$ is the NE associated with $x^\lambda$ so that $x^\lambda\to x^*$ as $\lambda\to 0$. From the above we see that
\begin{align}
\frac{\BRhat_i(x^\lambda)}{\partial x_j} = & \frac{1}{\lambda}x_i^\lambda\frac{\partial U(a_i^1,x_{-i}^\lambda)}{\partial x_j}\\
& - \frac{1}{\lambda}x_i^\lambda\left(x_i^\lambda\frac{\partial U(a_i^1,x_{-i}^\lambda)}{\partial x_j} + (1-x_i^\lambda)\frac{\partial U(a_i^1,x_{-i}^\lambda)}{\partial x_j} \right)\\
& = \frac{1}{\lambda}x_i^\lambda(1-x_i^\lambda)\left(\frac{\partial U(a_i^1,x_{-i}^\lambda)}{\partial x_j} - \frac{\partial U(a_i^2,x_{-i}^\lambda)}{\partial x_j} \right).
\end{align}
From \eqref{eq:F-to-partialU} we see that
$$
\frac{\BRhat_i(x^\lambda)}{\partial x_j} = \frac{1}{\lambda}x_i^\lambda(1-x_i^\lambda)\frac{\partial^2 U(x^\lambda)}{\partial x_i\partial x_j}
$$
Without loss of generality, assume that the set of players is ordered so that players $1,\ldots,\tilde N$ play mixed strategies.
Let $R_\lambda$ be the diagonal $\tilde N\times \tilde N$ matrix given by
$$
R_{\lambda} := \myDiag(x_1^\lambda(1-x_1^\lambda),\ldots,x_{\tilde N}^\lambda(1-x_{\tilde N}^\lambda)).
$$
%(Note that the fixed point $x$ of $\BRhat$ in question varies with $\lambda.$)
Define the matrices
$$
B_\lambda := \left( x_i^\lambda(1-x_i^\lambda)\frac{\partial^2 U(x^\lambda)}{\partial x_i\partial x_j}\right)_{\substack{i=1,\ldots,\tilde N\\ j=\tilde N+1,\ldots,N}},
$$
and
$$
C_\lambda := \left( x_i^\lambda(1-x_i^\lambda)\frac{\partial^2 U(x^\lambda)}{\partial x_i\partial x_j}\right)_{i,j=\tilde N+1,\ldots,N},
$$
and observe that we have
\begin{equation} \label{eq:Jacobian}
DF(x^\lambda) = \frac{1}{\lambda}
\begin{pmatrix}
R_\lambda H(x^\lambda) & B_\lambda\\
B_\lambda^T & C_\lambda
\end{pmatrix} - I_N,
\end{equation}
where
%$0$ above denotes the zero matrix of appropriate dimension and
$I_N$ is the $N\times N$ identity matrix, and $H(x^\lambda)$ (defined w.r.t. $x^*$) is defined in \eqref{eq:def-H}.

Without loss of generality, assume that the pure strategy set is ordered so that $x_i^* > 0$ for all $i$.
Since $\Gamma$ is regular, $x^*$ is a strict NE and there exists constants $c_1$ and $c_2$ such that $U(a_i^1,x_{-i}) \geq c_1 > c_2 \geq U(a_i^2,x_{-i})$ for all $x$ in a ball about $x^*$. For $i > \tilde N$, by our ordering of the pure strategy set we have $x_i^* = 1$, and, since $x_i^\lambda = \BRhat_i(x^\lambda)$, using \eqref{eq:BR-def} we have $\frac{1}{\lambda}(1-x_i^\lambda)\to 0$ and $x_i^\lambda\to 1$ as $\lambda \to \infty$. Thus, $\frac{1}{\lambda}B_\lambda$ and $\frac{1}{\lambda}C_\lambda$ converge entrywise to zero.

If $x^*$ is a pure NE, then we have $DF(x^\lambda) = \frac{1}{\lambda}C_\lambda - I_N$, and taking $\lambda>0$ sufficiently small, each eigenvalue of $DF(x^\lambda)$ may be brought arbitrarily close to -1. Thus, for all $\lambda>0$ sufficiently small, $x^\lambda$ is hyperbolic and stable under \eqref{eq:BR-perturbed}.

If $x^*$ is a mixed NE, then, since $\Gamma$ is regular, $H(x^*)$ is invertible and has at least one positive eigenvalue. Since $H(x^\lambda) \to H(x^*)$ as $\lambda \to 0$, taking $\lambda\to 0$ we see that $DF(x^\lambda)$ is invertible and has at least one positive eigenvalue for all $\lambda>0$ sufficiently small (\hspace{-.01em}\cite{kato2013perturbation}, Theorem II.6.1). Thus, for all $\lambda>0$ sufficiently small, $x^\lambda$ is hyperbolic and unstable under \eqref{eq:BR-perturbed}.
\end{proof}

\subsection{Main Result}
Finally, we state the main result of the paper.
\begin{theorem} \label{thrm:main-result}
In almost all $N\times 2$ potential games, for all $\lambda>0$ sufficiently small, smooth FP converges to a pure-strategy Nash distribution with probability 1.
\end{theorem}

Theorem \ref{thrm:main-result} follows immediately from Theorem \ref{thrm:sandholm}, Theorem \ref{thrm:reg-pot-games}, and Proposition \ref{prop:ND-hyperbolic}.

\begin{remark}
We remark that by Theorem \ref{thrm:ND-to-NE}, Theorem \ref{thrm:main-result} implies that any limit point of smooth FP may be brought arbitrarily close to the set of pure-strategy NE by taking $\lambda$ sufficiently small.
\end{remark}

\section{Conclusions} \label{sec:conclusion}
Game-theoretic learning dynamics are typically known to converge to the set of NE in potential games. However, the question of convergence to pure vs mixed strategies is often unclear in algorithms that evolve in the mixed strategy space. In this paper we have considered the case of smooth FP in $N\times 2$ potential games and shown convergence to the neighborhood of pure NE with probability 1. The key enabler of our result and analysis technique was notion of a \emph{regular potential game} \cite{swenson2017regular}. We hypothesize that analogous results may be obtained for other learning algorithms in potential games by applying similar techniques.
\bibliographystyle{IEEEtran}
\bibliography{myRefs}

\end{document}